\documentclass[conference]{IEEEtran}
\usepackage{cite,graphicx,amssymb,amsmath,color,textcomp,epsfig}
\usepackage{extarrows,multirow,multicol}
\usepackage{mathabx} 
\usepackage{cleveref} 
\usepackage{bm}
\usepackage{booktabs}
\usepackage{subeqnarray}
\usepackage{subfig}
\usepackage{float}
\usepackage{enumerate}
\usepackage{algorithm}
\usepackage{multicol}
\usepackage{bbm}
\usepackage{algpseudocode}

\usepackage{caption}
\newtheorem{remark}{\underline{Remark}}
\newtheorem{lemma}{Lemma}

\newlength{\figwidth}
\IEEEoverridecommandlockouts
\begin{document}
\bstctlcite{IEEEexample:BSTcontrol}
\title{
Performance Analysis of Flexible Duplex Inter-Satellite Links in LEO Networks \vspace{-0.1cm}}
\author{
\IEEEauthorblockN{Yomali~Lokugama, Charith~Dissanayake, Saman~Atapattu, and 
Kandeepan~Sithamparanathan}
 \IEEEauthorblockA{
Department of Electrical and Electronic Engineering, RMIT University, Victoria, Australia\\
\IEEEauthorblockA{Email:\{yomali.lokugama,\,charith.dissanayake,\,saman.atapattu,\,kandeepan.sithamparanathan\}@rmit.edu.au}\,
}
\vspace{-1cm}

}
\maketitle

\begin{abstract}
This paper investigates energy-efficient inter-satellite communication in Low Earth Orbit (LEO) networks, where satellites exchange both buffered and newly generated data through half-duplex inter-satellite links (ISLs). Due to orbital motion and interference-prone directional asymmetry, the achievable ISL capacities in opposite directions vary dynamically, leading to inefficient utilization under conventional fixed or alternating duplex modes. To address this, we propose a \emph{Flexible Duplex (FlexD)} scheme that adaptively selects the ISL transmission direction in each slot to maximize instantaneous end-to-end sky-to-ground throughput, jointly accounting for ISL quality, downlink conditions, and queue backlogs. A unified analytical framework is developed that transforms the bottleneck rate structure into an equivalent SINR domain, enabling closed-form derivations of throughput outage probability and energy efficiency under deterministic ISLs and Rician satellite-to-ground fading. The analysis reveals distinct operating regions governed by ISL and backlog constraints and provides tractable bounds for ergodic rate and energy efficiency. Numerical results confirm that FlexD achieves higher reliability and up to 30\% improvement in energy efficiency compared with conventional half- and full-duplex schemes under realistic inter-satellite interference conditions.
\end{abstract}

\section{Introduction}
As sixth-generation (6G) networks evolve toward global coverage, Low Earth Orbit (LEO) satellite constellations are emerging as a critical enabler for ubiquitous connectivity and low-latency backhaul~\cite{union2022future}. Operating in large-scale formations, each satellite maintains an independent footprint and cooperates with neighboring satellites via inter-satellite links (ISLs) to support multi-hop routing, data relaying, and service continuity~\cite{radhakrishnan2016survey,10937133}.
By leveraging ISLs, satellites can forward buffered or real-time data without immediate ground contact, thereby enhancing spatial coverage and reducing reliance on terrestrial gateways. However, the performance of ISLs is inherently dynamic. Time-varying interference, orbital motion, and propagation geometry cause large fluctuations in the signal-to-interference-plus-noise ratio (SINR), leading to asymmetric link conditions between neighboring satellites~\cite{Cao_Qihang}. At the same time, onboard energy remains a scarce resource, making energy efficiency (EE) a fundamental design objective in LEO networks. Despite extensive work on ISL scheduling and power optimization, the joint impact of link asymmetry, inter-satellite interference (ISI), and data availability on end-to-end performance has received limited analytical attention.

Existing studies can broadly be classified into two categories: system-level optimization and analytical modeling. System-level approaches such as~\cite{Enhancingearth,leyva2021inter,7572177,li2014power,li2020energy,Pi_Jiahao} aim to improve throughput or EE through resource allocation and ISL scheduling. For example, \cite{Enhancingearth} optimizes ISL assignment for throughput enhancement but neglects ISI, while~\cite{leyva2021inter} considers ISI in the scheduling design without analytical characterization. Similarly,~\cite{7572177} and~\cite{li2014power} investigate EE and power allocation under Quality-of-service (QoS) and satellite power constraints, but they disregard the interaction between ISLs and satellite-ground links (SGLs).
To bridge this gap,~\cite{li2020energy} incorporates cooperative ISL-SGL transmissions for EE maximization, whereas~\cite{Pi_Jiahao} applies multi-agent reinforcement learning for ISL link planning. Yet, these frameworks remain optimization-centric, lacking closed-form performance that can provide deeper physical insight.
Analytical works in~\cite{10891202,10622540,9632432} study ISL/LEO systems from a fundamental perspective:~\cite{10891202} analyzes ISL under distance uncertainty without ISI,~\cite{10622540} mitigates interference via beamforming in LEO networks but ignores ISI, and~\cite{9632432} models joint ISL-SGL throughput with static duplexing.
Across all these works, both ISL and SGL links are constrained to operate in either half-duplex (HD) or full-duplex (FD) modes~\cite{10804611}, which are inefficient under dynamic directional asymmetry.

\begin{figure*}[t]
    \centering
    \subfloat[Constellation configuration at $t=\tau$.]{\includegraphics[width=0.33\textwidth]{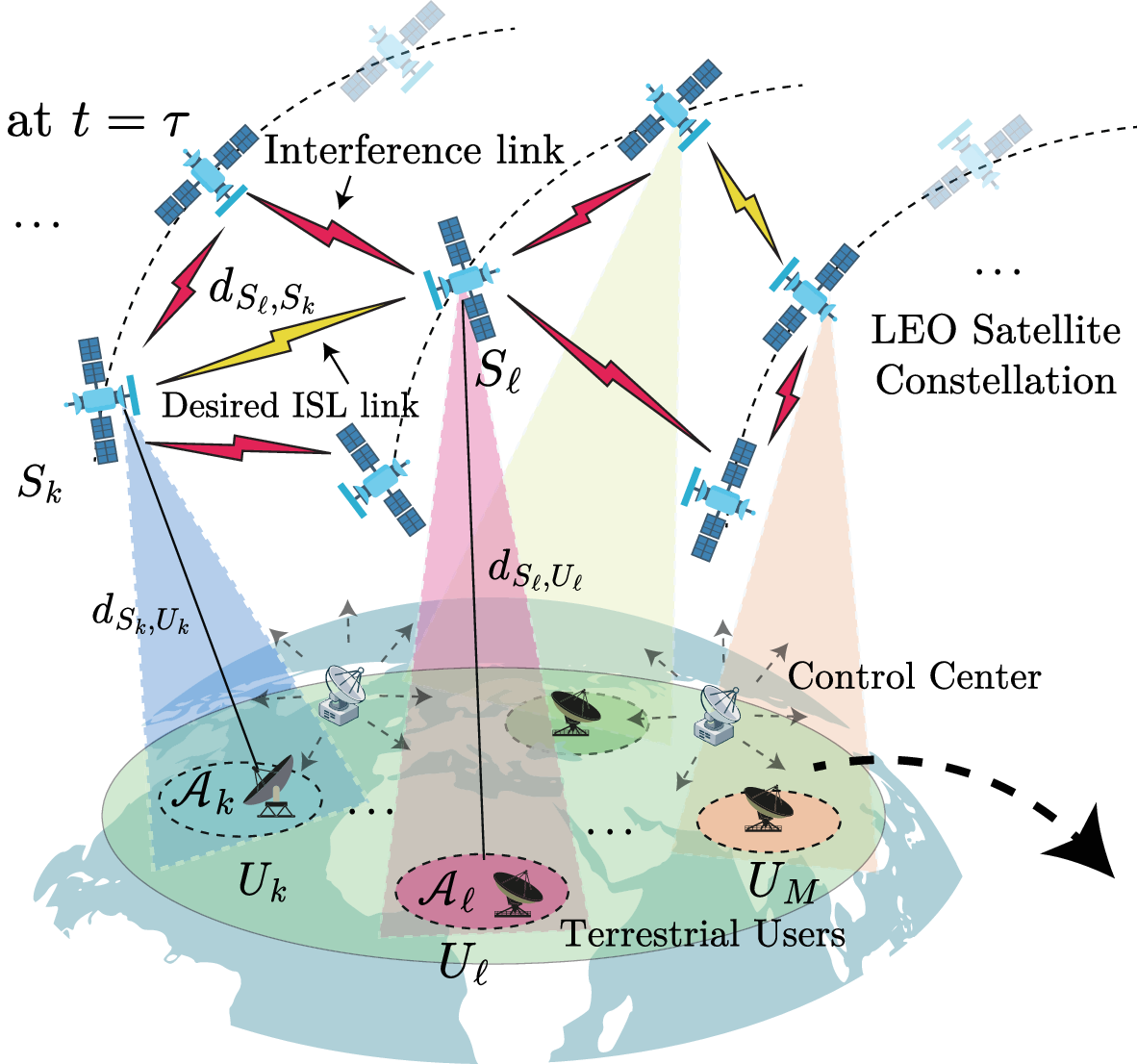}\label{fig_all_network}} 
    \hfill
    \subfloat[Constellation after handover at $t=T_{\rm cov}+\tau$.]{\includegraphics[width=0.33\textwidth]{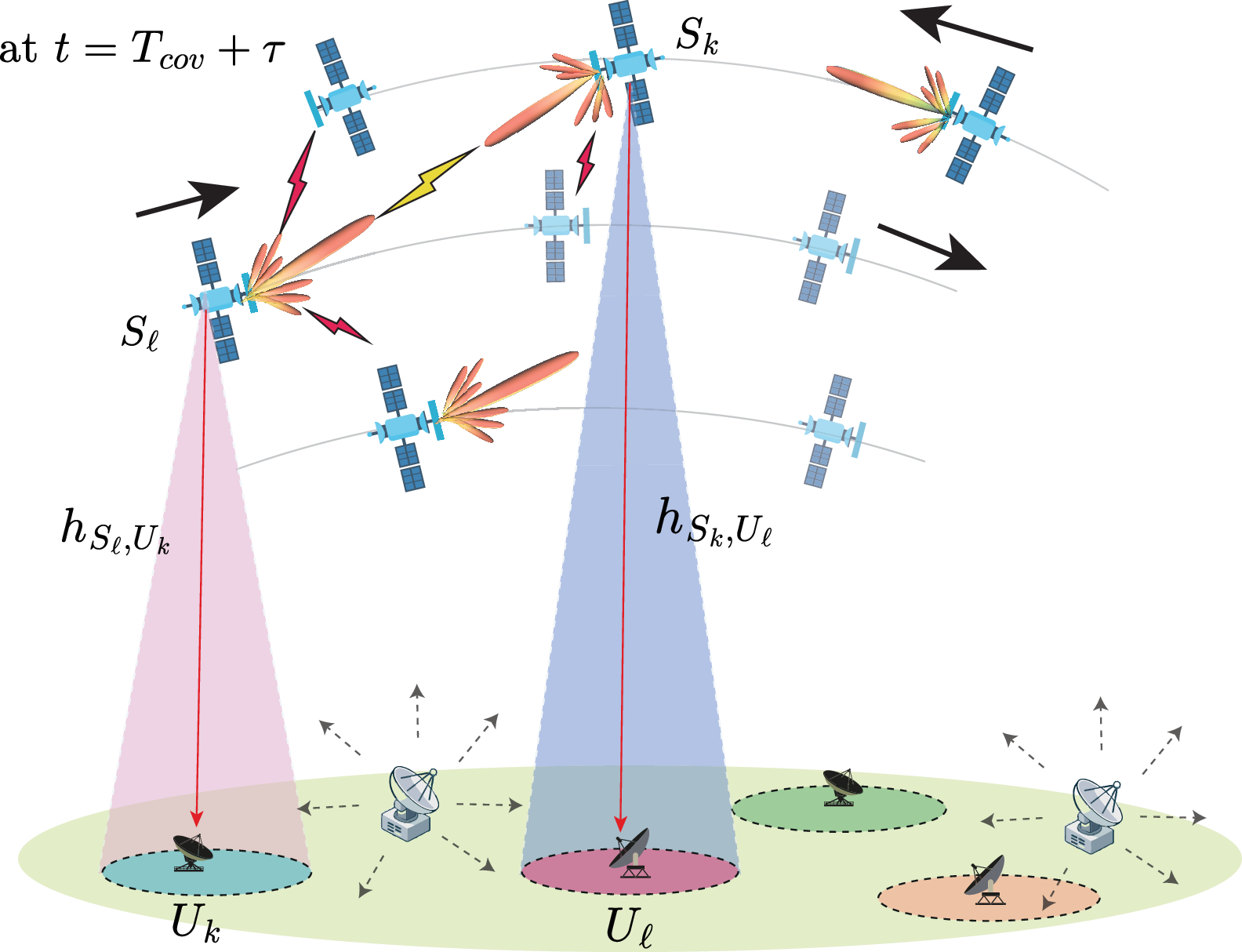}\label{fig:tcov1}} 
    \hfill
    \subfloat[Temporal variation of ISL  for $S_k$ and $S_\ell$.]{\includegraphics[width=0.33\textwidth]{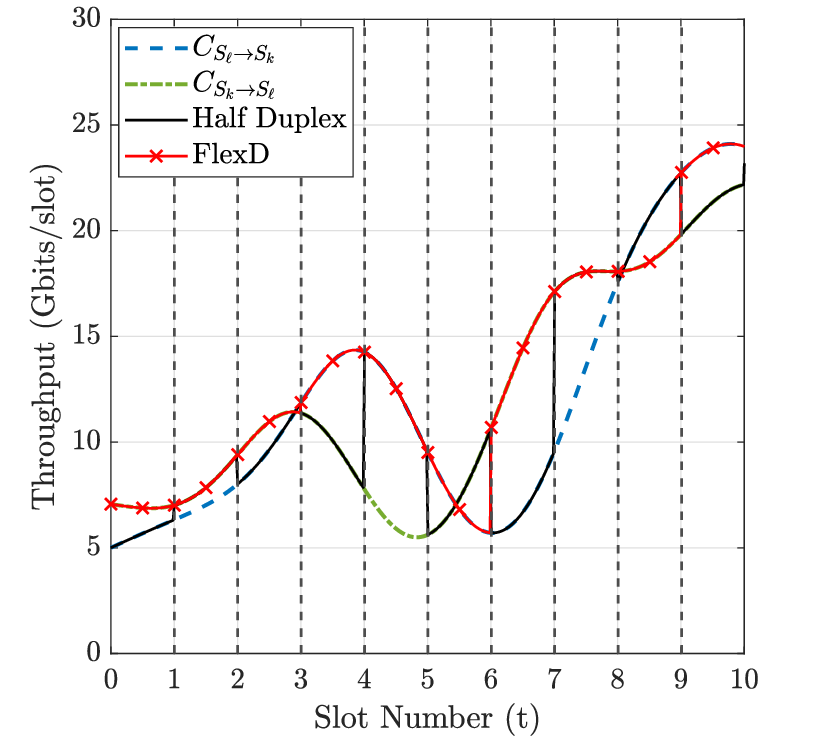}\label{fig:mot}} 
    \caption{System architecture and motivation for the proposed FlexD scheme: (a) initial constellation configuration, (b) post-handover configuration, and (c) comparison of ISL capacities $C_{S_k\to S_\ell}$ and $C_{S_\ell\to S_k}$ under FlexD and conventional HD.}
    \vspace{-4.5mm}
\end{figure*}

In contrast, the concept of \emph{flexible duplex (FlexD)}-recently explored in terrestrial multi-user systems-allows slot-wise adaptation of transmission direction according to instantaneous channel and interference conditions, thereby improving throughput and EE in interference-limited networks~\cite{11196944,Lokugama2025_vtcfall}.
However, its application to satellite systems, particularly to coupled ISL-SGL architectures under mobility and interference, remains unexplored. 
\textit{This work presents the first analytical FlexD framework for LEO satellite networks incorporating ISL-SGL coupling and inter-satellite interference.}

The main contributions are summarized as follows: (i) A novel end-to-end communication model is formulated, where deterministic ISLs collaborate with random fading SGLs to deliver backlogged and new data under dynamic mobility.
(ii) A FlexD strategy is introduced to dynamically select the optimal ISL  direction per slot, maximizing instantaneous end-to-end throughput.
(iii) Closed-form expressions are derived for throughput outage probability and EE, based on a unified SINR-domain of the ISL-downlink bottleneck.
(iv) Numerical results validate the theory, demonstrating that FlexD achieves significantly higher reliability and up to 30\% improvement in EE compared to conventional HD and FD systems under ISI.


\section{System Model}
\label{sec:system}
This section describes the LEO constellation model, emphasizing the coupling between ground coverage and inter-satellite links (ISLs). We first define the constellation geometry and temporal structure, then introduce the two-hop data delivery induced by mobility and buffered traffic, followed by the half-duplex ISL scheduling principle (FlexD).
\subsection{Constellation, Geographic Partition, and Time}
\label{sec:constellation_time}

Consider a LEO constellation comprising  $M$ satellites,
\(
\mathcal{S}=\{S_1,S_2,\dots,S_M\},
\)
serving $M$ non-overlapping regions: 
\[
\mathcal{A}=\{\mathcal{A}_1,\mathcal{A}_2,\dots,\mathcal{A}_M\}, 
\quad \mathcal{A}_i\cap\mathcal{A}_j=\emptyset,\ \forall\, i\neq j,
\]
such that $\bigcup_{i=1}^M\mathcal{A}_i$ spans the entire service theater.  
Each region $\mathcal{A}_i$ contains a representative ground user or gateway $U_i$ that acts as the traffic endpoint.  
As illustrated in Fig.~\ref{fig_all_network}, satellites $S_k$ and $S_\ell$ serve the corresponding regions $\mathcal{A}_k$ and $\mathcal{A}_\ell$, where the ground nodes $U_k$ and $U_\ell$ reside.
A region remains continuously visible to one satellite for a \emph{coverage window} of duration $T_{\rm cov}$, during which the region-satellite association is stable.  
Time is divided into discrete slots indexed by $t\in\mathbb{Z}_{\ge0}$, each of duration $T_{\rm slot}\!\ll\!T_{\rm cov}$ such that geometry, visibility, and channel gains remain quasi-static within a slot.  
Hence, multiple scheduling slots fit within a single coverage window.

At the start of a window ($t=\tau$), the regions $\mathcal{A}_k$ and $\mathcal{A}_\ell$ are served by satellites $S_k$ and $S_\ell$, respectively.  
Due to orbital motion, after approximately $T_{\rm cov}$ seconds these satellites move out of visibility, and service is handed over to adjacent satellites entering the view of $\mathcal{A}_k$ and $\mathcal{A}_\ell$.  
This mutual exchange is denoted by $S_k\!\leftrightarrow\!S_\ell$ at $t=\tau+T_{\rm cov}$, as shown in Fig.~\ref{fig:tcov1}.  
The handover is not instantaneous: the departing satellites $S_k$ and $S_\ell$ may still buffer undelivered data for $U_k$ and $U_\ell$.  
Therefore, an inter-satellite data transfer between $(S_k,S_\ell)$ is required to complete the pending transmissions during or immediately after the handover.

\subsection{Communication Flows and the Two-Hop Delivery Problem}
\label{sec:comm_flows}

During a scheduling time $t=\tau$, the primary downlinks are the one-hop transmissions 
$S_k\!\to\!U_k$ and $S_\ell\!\to\!U_\ell$ (see Fig.~\ref{fig_all_network}), where 
$D_k(t)\!\ge\!0$ and $D_\ell(t)\!\ge\!0$ denote their instantaneous capacities (bits/slot) determined by the PHY-layer model in Sec.~\ref{sec:signalmodel}.  

At the end of a coverage window, satellites may exchange their serving regions. 
If $S_k$ and $S_\ell$ subsequently serve $U_\ell$ and $U_k$, respectively, each may still hold undelivered data from the previous window. 
These residual packets, of sizes $Q_{S_k\to U_k}\!\ge\!0$ and $Q_{S_\ell\to U_\ell}\!\ge\!0$ (bits), must be relayed over the inter-satellite link (ISL), giving rise to two-hop routes
\[
S_k\!\to\!S_\ell\!\to\!U_k, 
\qquad 
S_\ell\!\to\!S_k\!\to\!U_\ell.
\]
For practical and analytical convenience, the backlog terms $Q_{S_k\to U_k}$ and $Q_{S_\ell\to U_\ell}$ are modeled as deterministic quantities representing the available buffered data, while their slower stochastic evolution due to external arrivals is left outside the present scope.  
Two-hop relaying can also occur \emph{within} a coverage window due to persistent cross-satellite transfers, asymmetric link efficiencies, or new arrivals at non-serving satellites.  
Without loss of generality, we consider the representative satellite pair $(S_k,S_\ell)$ exchanging data within the same window and assume sufficient backlog in both directions.

Let $C_{S_k\to S_\ell}(t)\!\ge\!0$ and $C_{S_\ell\to S_k}(t)\!\ge\!0$ denote the instantaneous ISL capacities (bits/slot) in slot~$t$.  
If packets for $U_\ell$ are still buffered at $S_\ell$, the achievable end-to-end two-hop rate over $S_k\!\to\!S_\ell\!\to\!U_k$ and $S_\ell\!\to\!S_k\!\to\!U_\ell$ are constrained by the bottleneck as
\begin{align}
R_{S_k\to U_k}(t)&=
\min\!\big\{\,C_{S_k\to S_\ell}(t),\,D_k(t),\,Q_{S_k\to U_k}(t)\big\} \label{eq:bottleneck_metric1b},\\
R_{S_\ell\to U_\ell}(t)&=
\min\!\big\{\,C_{S_\ell\to S_k}(t),\,D_\ell(t),\,Q_{S_\ell\to U_\ell}(t)\big\}.
\label{eq:bottleneck_metric1a}
\end{align}
Each expression represents the instantaneous half-duplex end-to-end service rate achievable through the ISL-downlink chain in the respective direction.

\begin{remark}[ISL Availability]
The analysis presumes an active bidirectional ISL between $S_k$ and $S_\ell$ during the considered window.  
Practical impairments such as intermittent visibility or pointing loss are mitigated by higher-layer routing and scheduling mechanisms; the present study focuses on the instantaneous PHY-layer behavior when the ISL is available.
\end{remark}


\section{Proposed Flexible Duplex (FlexD) for ISLs}
\label{sec:flexd_principle}

ISLs serve as shared backhaul between neighboring satellites and operate in half-duplex mode on a given band; that is, for any pair $(S_k,S_\ell)$, only one transmission direction can be active per slot.

\subsection{Directional ISL Asymmetry and Motivation}
The instantaneous ISL capacity depends on direction-specific factors such as antenna pointing, residual carrier frequency offset (CFO), relative distance, and interference at the receiving satellite~\cite{Cao_Qihang}.  Let \(\mathcal{I}(t)\) denote the set of satellites that are concurrently transmitting on the same resource block in slot \(t\) and are visible to either \(S_k\) or \(S_{\ell}\). The subsets \(\mathcal{I}_{S_k}(t) \subseteq \mathcal{I}(t)\) and \(\mathcal{I}_{S_\ell}(t) \subseteq \mathcal{I}(t)\) denote the sets of interfering satellites observed by \(S_k\) and \(S_\ell\), respectively.    
Since $\mathcal{I}_{S_k}(t)\!\neq\!\mathcal{I}_{S_\ell}(t)$ in general due to differing spatial geometry, field-of-view (FOV), and scheduling, the directional ISL capacities $C_{S_k\to S_\ell}(t)$ and $C_{S_\ell\to S_k}(t)$ are typically asymmetric and time-varying even within a single coverage window.

Fig.~\ref{fig:mot} illustrates the temporal variation of the two ISL capacities $C_{S_k\to S_\ell}(t)$ and $C_{S_\ell\to S_k}(t)$.  
The forward and reverse capacities show strong slot-level asymmetry and distinct long-term trends.  
Conventional half-duplex operation, which alternates direction in a fixed or pre-assigned pattern slot-wise, cannot efficiently exploit these variations.  
This motivates an adaptive slot-by-slot selection mechanism that always utilizes the stronger direction.  
Furthermore, because the end-to-end service rate also depends on the downlink capacity and available backlog ($D_k$, $D_\ell$, $Q_{S_k\to U_k}$, $Q_{S_\ell\to U_\ell}$), opportunistic direction selection can substantially enhance the overall two-hop throughput-forming the basis of the proposed \emph{Flexible Duplex (FlexD)} scheme.

\subsection{FlexD Direction Selection}
The instantaneous two-hop service rates (bits/slot) for the two flows $S_k\to U_k$ and $S_\ell\to U_\ell$ are given by $R_{S_k\to U_k}(t)$ and $R_{S_\ell\to U_\ell}(t)$ from \eqref{eq:bottleneck_metric1b} and \eqref{eq:bottleneck_metric1a}, representing the effective ISL-downlink rates in each direction.  
FlexD dynamically selects the active half-duplex direction $d^\star(t)$ in each slot to maximize the instantaneous end-to-end throughput:
\begin{align}
d^\star(t)&=\arg\max_{d\in\{S_k\to U_k,\ S_\ell\to U_\ell\}} R_d(t),\\
R^\star(t)&=\max\{R_{S_k\to U_k}(t),\,R_{S_\ell\to U_\ell}(t)\}.
\label{eq:flexd_decision}
\end{align}
This per-slot rule allocates the ISL to the direction offering the highest achievable throughput. The network is coordinated by a  \emph{control center}, possessing environment awareness,
which has access to satellite and user locations. 

\begin{remark}[Relation to SINR]
If $C_{S_k\to S_\ell}(t)=W_{\rm isl}\log_2(1+\gamma_{S_k\to S_\ell}(t))$ and $C_{S_\ell\to S_k}(t)=W_{\rm isl}\log_2(1+\gamma_{S_\ell\to S_k}(t))$, maximizing $R_d(t)$ reduces to maximizing $\gamma_d(t)$ only when the ISL link is the bottleneck, i.e., $C(t)\!\le\!D, Q$ in \eqref{eq:bottleneck_metric1b}-\eqref{eq:bottleneck_metric1a}.  
Otherwise, downlink or backlog constraints dominate, and the full min-structure must be retained.  
The next section analyzes the SINR distributions and their effect on the FlexD rate.
\end{remark}

\section{SINR Analysis and Physical-Layer Model}
\label{sec:signalmodel}

\subsection{From Bottleneck Rates to SINR Equivalents}
For a satellite pair $(S_\ell,S_k)$, the per-slot two-hop service rate in the direction $S_\ell\!\to\!U_\ell$ is given by \eqref{eq:bottleneck_metric1a}
with $C_{S_\ell\to S_k}(t)$ is the ISL capacity, $D_\ell(t)$ the downlink capacity to user $U_\ell$, and $Q_{S_\ell\to U_\ell}(t)$ (bits) the backlog stored at $S_\ell$ for $U_\ell$.

The ISL and downlink capacities are related to their instantaneous SINRs by
\begin{align}
C_{S_\ell\to S_k}(t)&=W_{\rm isl}\log_2\!\big(1+\gamma_{S_\ell\to S_k}(t)\big),\\
D_\ell(t)&=W_{\rm dl}\log_2\!\big(1+\gamma_{S_k\to U_\ell}(t)\big),
\end{align}
where $W_{\rm isl}$ and $W_{\rm dl}$ are the respective bandwidths.  
In general, \eqref{eq:bottleneck_metric1a} cannot be written directly as a ``minimum over SINRs’’ unless (i) a common bandwidth is assumed and (ii) the backlog term is expressed as an equivalent rate.  
Assuming $W_{\rm isl}=W_{\rm dl}=W$, the backlog-equivalent SINR is defined as
\begin{equation}
\gamma_{Q_{S_\ell\to U_\ell}}^{(W)}(t)\triangleq 2^{Q_{S_\ell\to U_\ell}(t)/W}-1.
\label{eq:backlog}
\end{equation}
Then the rate in \eqref{eq:bottleneck_metric1a} can be equivalently expressed as
\[
R_{S_\ell\to U_\ell}(t)=
W\log_2\!\Big(1+\gamma_{S_\ell\to U_\ell}(t)\Big),
\]
with corresponding SINR term
\begin{align}
    \gamma_{S_\ell\to U_\ell}(t)\hspace{-0.1cm}=\hspace{-0.1cm}\min\{{\gamma_{S_\ell\to S_k}(t)},\;
{\gamma_{S_k\to U_\ell}(t)},\;
{\gamma_{Q_{S_\ell\to U_\ell}}^{(W)}(t)}\}\label{eq:rate_to_sinr_equiv1}.
\end{align}
The reverse direction rate $R_{S_k\to U_k}(t)$ follows by interchanging $k$ and $\ell$.  
If $W_{\rm isl}\neq W_{\rm dl}$, the rate-domain forms in \eqref{eq:bottleneck_metric1b}-\eqref{eq:bottleneck_metric1a} should be used directly for analysis.

Henceforth, for notational simplicity, we define the large-scale gains and average SNR between nodes $i$ and $j$ as 
\begin{align}
\label{pathloss_gain}
    \alpha_{i,j}(t)=\frac{G_{i,j}G_{j,i}c^2}{(4\pi f\,d_{i,j}(t))^2},\qquad  \bar{\gamma}_{i,j}(t)=\frac{P_i \alpha_{i,j}(t)}{\sigma^2_j}
\end{align}
respectively. Here, $G_{i,j}$ denotes the directional antenna gain (e.g., UPA/ULA, including Doppler-compensated pointing)~\cite{balanisbook,dissanayake2025uniform}, 
$d_{i,j}(t)$ the euclidean distance between node $i$ and $j$, $f$ the carrier frequency, $c$ the speed of light, $P_i$ is the transmit power of node $i$ and $\sigma^2_j$ is the receiver noise power at node $j$ modeled as AWGN $\sim \mathcal{CN}(0,\sigma_{j}^2)$.  
\subsection{Deterministic ISL Link SINR: $S_\ell\!\to\!S_k$}
The received SINR at $S_k$ in slot $t$ is
\begin{equation}
\gamma_{S_\ell\to S_k}(t)=
\frac{\bar{\gamma}_{S_\ell,S_k}(t)}
{\displaystyle\sum_{S_n\in\mathcal{I}_{S_k}(t)}\bar{\gamma}_{S_n,S_k}(t)+1},
\label{eq:isl_sinr}
\end{equation}
where $\bar{\gamma}_{S_\ell,S_k}(t)$ and $\bar{\gamma}_{S_n,S_k}(t)$ obtained from \eqref{pathloss_gain} and $\mathcal{I}_{S_k}(t)$ is the co-channel interferers set visible to $S_k$ at slot $t$. 

\subsection{Random Downlink SINR: $S_k\!\to\!U_\ell$}
The received SINR at $U_\ell$ in slot $t$ is
\begin{equation}
\gamma_{S_k\to U_\ell}(t)=
\bar{\gamma}_{S_k,U_\ell}(t)\,|h_{S_k,U_\ell}(t)|^2,
\label{eq:dl_sinr}
\end{equation}
where $\bar{\gamma}_{S_k,U_\ell}(t)$ obtained from \eqref{pathloss_gain} and $h_{S_k,U_\ell}(t)$ represents the  small-scale Rician fading $\sim\mathcal{CN}(\mu,\sigma_g^2)$  which accurately captures the directional satellite-ground propagation~\cite{Zhao_Meihui}. Thus the CDF of $\gamma_{S_k\to U_\ell}$ is given by
\begin{align}\label{eq:cdf_rice}
F_{X}(\bar{\gamma}_{S_k,U_\ell};x)&=1-Q_1\!\left(\tfrac{|\mu|}{\sqrt{\sigma_g^2/2}},
\sqrt{\tfrac{2x}{\bar{\gamma}_{S_k,U_\ell}\sigma_g^2}}\right),
\end{align}
where $Q_1(\cdot,\cdot)$ denotes the first-order Marcum $Q$-function. 
Similarly the CDF of $\gamma_{S_\ell\to U_k}$ can be written as $F_{X}(\bar{\gamma}_{S_\ell,U_k};x)$. 
 
Corresponding SINR terms for the reverse direction $S_k\!\to\!U_k$ follows analogously by index  $k$ and $\ell$ interchange.
\section{FlexD System Performance Analysis}
The performance is evaluated under two primary metrics (i) Throughput Outage Probability, and (ii) Energy Efficiency.
\subsection{Throughput Outage Probability}
\label{sec:non_rec}

The throughput outage probability quantifies the likelihood that the instantaneous two-hop rate falls below a target~$\delta$.  
With a pre-log factor of~${1}/{2}$ due to HD operation, it is defined as
\begin{equation}
\label{outage_def}
P_{\mathrm{o}}=\Pr\!\left(\tfrac{W}{2}\log_2(1+\Gamma)\le\delta\right),
\end{equation}
where $\zeta\!\triangleq\!2^{(2\delta)/W}-1$ is the equivalent SINR threshold, and
\begin{equation}
\label{system_SINR}
\Gamma=\max\{\gamma_{S_\ell\to U_\ell}(t),\,\gamma_{S_k \to U_k}(t)\},
\end{equation}
with composite SINR terms $\gamma_{S_\ell\to U_\ell}(t)$ and $\gamma_{S_k\to U_k}(t)$ can be found using \eqref{eq:rate_to_sinr_equiv1} and interchanging indexes $k,\ell$ respectively. 

\begin{lemma}[Throughput Outage Probability]
\label{lem:outage}
The system outage probability $P_{\mathrm{o}}=\Pr(\Gamma\le\zeta)$ is expressed as
\begin{align}
\label{eq:CDF_gamma_star}
P_{\rm o}(\zeta)\hspace{-0.1cm}=\hspace{-0.1cm}
\begin{cases}
F_{X}(\bar{\gamma}_{S_k,U_\ell};{\zeta})F_{X}(\bar{\gamma}_{S_\ell,U_k};{\zeta}),\hspace{-0.1cm} & \zeta < \Omega_{\min}, \\[2pt]
F_{X}(\bar{\gamma}_{S_\ell,U_k};{\zeta}), & \hspace{-0.1cm}\Omega_{\ell}\le \zeta < \Omega_{k}, \\[2pt]
F_{X}(\bar{\gamma}_{S_k,U_\ell};{\zeta}), & \hspace{-0.1cm}\Omega_{k}\le \zeta < \Omega_{\ell}, \\[2pt]
1, & \hspace{-0.1cm}\zeta \ge \Omega_{\max},
\end{cases}
\end{align}
where $F_X(\cdot,\cdot)$ denotes the Rician CDF defined in~\eqref{eq:cdf_rice}.  
The deterministic cut levels are
\[
\Omega_{\ell}=\min\{\gamma_{S_\ell\to S_k},\,\gamma_{Q_{S_\ell\to U_\ell}}^{(W)}\},\,
\Omega_{k}\hspace{-0.1cm}=\hspace{-0.1cm}\min\{\gamma_{S_k\to S_\ell},\,\gamma_{Q_{S_k\to U_k}}^{(W)}\},
\]
and $\Omega_{\min}=\min(\Omega_{\ell},\Omega_{k})$, $\Omega_{\max}=\max(\Omega_{\ell},\Omega_{k})$.
\end{lemma}

\begin{IEEEproof}[Proof Sketch]
Let $Y=\gamma_{S_\ell\to U_\ell}=\min\{\Omega_{\ell},H_\ell\}$ and 
$Z=\gamma_{S_k\to U_k}=\min\{\Omega_{k},H_k\}$, 
so that $\Gamma=\max\{Y,Z\}$.  
For any $\zeta\ge0$,
$\Pr(\Gamma\le\zeta)=\Pr(Y\le\zeta,\,Z\le\zeta)
=F_Y(\zeta)F_Z(\zeta),$ since $Y$ and $Z$ are independent. The CDF of $Y$ is given as 
\begin{align}
F_Y(\zeta)&=
\begin{cases}
F_{X}(\bar{\gamma}_{S_k,U_\ell};\zeta), & \zeta<\Omega_{\ell},\\
1, & \zeta\ge\Omega_{\ell},
\end{cases}
\end{align}
Similarly, $F_Z(\zeta)$ can be found and the product of $F_Y(\zeta)$ and $F_Z(\zeta)$ yields \eqref{eq:CDF_gamma_star}.  
Finally, substituting $\zeta=2^{(2\delta)/W}-1$ maps the throughput threshold $\delta$ to the equivalent SINR threshold.
\end{IEEEproof}

\begin{remark}[Interpretation of Lemma~1 under No-Backlog Conditions]
When both inter-satellite exchange queues are empty, i.e., $Q_{S_\ell\to U_\ell}=Q_{S_k\to U_k}=0$, no relay traffic exists between $S_k$ and $S_\ell$. In this case, $R_{S_\ell\to U_\ell}(t)=R_{S_k\to U_k}(t)=0
\Longleftrightarrow
\gamma_{S_\ell\to U_\ell}(t)=\gamma_{S_k\to U_k}(t)=0$ and the ISL remains idle. The system therefore reduces to a purely \emph{direct downlink} mode, with instantaneous one-hop capacities
$
D_k(t)\hspace{-0.1cm}=\hspace{-0.1cm}W\hspace{-0.1cm}\log_2\,\!(1\hspace{-0.1cm}+\gamma_{S_\ell\to U_k}(t)),\hspace{-0.05cm}$ and $
D_\ell(t)\hspace{-0.1cm}=\hspace{-0.1cm}W\hspace{-0.1cm}\log_2\,\!(1\hspace{-0.1cm}+\gamma_{S_k\to U_\ell}(t)).
$
Substituting these into Lemma~1 yields decoupled per-user outage probabilities, $
P_{{\rm o},k}(\zeta)=F_{X}(\bar{\gamma}_{S_\ell,U_k};{\zeta}),\,
P_{{\rm o},\ell}(\zeta)=F_{X}(\bar{\gamma}_{S_k,U_\ell};{\zeta}).$ 
In practice, such links are excluded from the FlexD set during no-backlog slots, or equivalently modeled via an activation indicator $R_d(t)=\mathbbm{1}\{Q_d(t)>0\}\min\{\cdot\}$ to capture traffic-dependent ISL participation.

\end{remark}

\subsection{Energy Efficiency (EE)}
Energy efficiency (EE) measures how effectively a system converts power into transmitted information, defined as the ratio of ergodic throughput $\bar{R}$ to total power $P_T$:

\begin{equation}
\mathrm{EE} = 
{\bar{R}}/{P_{\mathrm{T}}},
\label{eq:EE_def}
\end{equation}
which is given in closed form in the following lemma, where $\mathrm{EE}$ is approximated as

\begin{lemma}[Energy Efficiency]
\label{lem:ee}
The system energy efficiency (EE) is approximated by
\begin{equation}
\mathrm{EE} \approx \frac{W}{2P_{\mathrm{T}}}
\log_2\!\big(1+\mathbb{E}[\Gamma]\big),
\label{eq:Rbar_approx}
\end{equation}
where the average system SINR $\mathbb{E}[\Gamma]$ is
\begin{equation}
\label{eq_sinr_energy}
\mathbb{E}[\Gamma]=
\begin{cases}
\mathcal{F}+\mathcal{G}(A_{S_k,U_\ell},\Omega_{\min},\Omega_{\max}), & \Omega_{\ell}\!\ge\!\Omega_{k},\\[2pt]
\mathcal{F}+\mathcal{G}(A_{S_\ell,U_k},\Omega_{\min},\Omega_{\max}), & \text{otherwise},
\end{cases}
\end{equation}
where $\mathcal{F}$ and $\mathcal{G}(\cdot)$ are given in~(22)-(24) and terms
\[A_{S_k,U_\ell}={1}/{\bar{\gamma}_{S_k, U_\ell}\sigma_g^2}, \quad
A_{S_\ell,U_k}={1}/{\bar{\gamma}_{S_\ell, U_k}\sigma_g^2}.\]
The slot-wise constants $\Omega_{\ell},\Omega_{k}$ and $\Omega_{\max},\Omega_{\min}$ are as defined in Lemma~\ref{lem:outage}. Substituting \eqref{eq_sinr_energy} into \eqref{eq:Rbar_approx} yields a tight analytical upper bound for $\mathrm{EE}$ (see Fig.~\ref{fig_all}).
\end{lemma}
\begin{IEEEproof}[Proof Sketch]
The system’s ergodic throughput, derived by taking the expectation $\mathbb{E}[\cdot]$ over \eqref{eq:flexd_decision} with \eqref{system_SINR}, as
\begin{equation}
\bar{R}\hspace{-0.1cm} =\hspace{-0.1cm}  \mathbb{E}[R^\star] 
\hspace{-0.1cm} = \hspace{-0.1cm} \mathbb{E}[\tfrac{W}{2}\hspace{-0.1cm}\log_2(1\hspace{-0.05cm}+\hspace{-0.05cm}\Gamma)]\hspace{-0.1cm} =\hspace{-0.2cm}  \int_{0}^{\infty} \hspace{-0.1cm} \frac{W}{2}\hspace{-0.1cm}\log_2(1 \hspace{-0.025cm}+ \hspace{-0.025cm}x) 
\hspace{-0.025cm}f_{\Gamma}(x)\hspace{-0.025cm} \, dx,
\label{eq:Rerg}
\end{equation} 
where $f_{\Gamma}(\cdot)$ is the PDF corresponding to the CDF of $\Gamma$ in \eqref{eq:CDF_gamma_star}.
\begin{figure*}[t]
\setcounter{equation}{23} 
\begin{align}
& \mathcal{F} \hspace{-0.1cm}= \Omega_{\min}\hspace{-0.1cm}
- e^{-\frac{2\mu^2}{\sigma_g^2}}\hspace{-0.1cm}\sum_{m=0}^{M} \sum_{n=0}^{M}\hspace{-0.1cm}
\frac{(\mu/\sigma_g)^{2mn}}{m!n!}
\!\left[
    \Omega_{\min} 
    \hspace{-0.1cm}- \frac{\Phi_m(A_{S_\ell, U_k}, \Omega_{\min})}{A_{S_\ell, U_k}}
    - \frac{\Phi_n(A_{S_k, U_\ell}, \Omega_{\min})}{A_{S_k, U_\ell}}
   \hspace{-0.1cm} +\hspace{-0.1cm} \Psi_{mn}(A_{S_\ell, U_k}, A_{S_k, U_\ell}, \Omega_{\min})
\right],
\tag{22}\\[-1mm]
&\mathcal{G}(\epsilon,\Omega_{\min},\Omega_{\max}) \hspace{-0.1cm}= \hspace{-0.1cm}e^{-\mu^2/\sigma_g^2}
\sum_{m=0}^{M}
\frac{(\mu/\sigma_g)^{2m}}{m!\,\epsilon}
\big[
    \Phi_m(\epsilon, \Omega_{\max})
    - \Phi_m(\epsilon, \Omega_{\min})
\big], \,
\Phi_m(A,S) = \frac{1}{A} 
\sum_{i=0}^{m} 
\left[1 - e^{-A S}\sum_{r=0}^{i}\frac{(A S)^r}{r!}\right],
\tag{23}\\[-1mm]
&\Psi_{mn}(A_{S_\ell, U_k},A_{S_k, U_\ell},S)
\hspace{-0.1cm}=\hspace{-0.15cm} \sum_{i=0}^{m}\sum_{j=0}^{n}\hspace{-0.05cm}
\frac{A_{S_\ell, U_k}^i A_{S_k, U_\ell}^j(i+j)!}{i!\,j!\,(A_{S_\ell, U_k}+A_{S_k, U_\ell})^{i+j+1}}
\!\left[1 - e^{-(A_{S_\ell, U_k}+A_{S_k, U_\ell})S}
\sum_{r=0}^{i+j}\frac{\big((A_{S_\ell, U_k}+A_{S_k, U_\ell})S\big)^r}{r!}\right].\tag{24}
\label{eq:F_G_defs}
\end{align}
\rule{\textwidth}{0.5pt}
\vspace{-4mm}
\end{figure*}
Since a closed-form evaluation of \eqref{eq:Rerg} is generally intractable, we employ Jensen’s inequality.  
Because $\log(1+x)$ is concave, 
\begin{equation}
\setcounter{equation}{21}
\mathbb{E}\!\left[\tfrac{W}{2}\log_2(1+\Gamma)\right] 
\le \tfrac{W}{2}\log_2\!\big(1+\mathbb{E}[\Gamma]\big).
\label{eq:Jensen}
\end{equation}
Let $Y=\gamma_{S_\ell\to U_\ell}=\min\{\Omega_{\ell},H_\ell\}$ and 
$Z=\gamma_{S_k\to U_k}=\min\{\Omega_{k},H_k\}$, so that $\Gamma=\max\{Y,Z\}$.  
Using the tail-sum representation,
\(
\mathbb{E}[\max(Y,Z)] = \int_{0}^{\infty}\Pr(\max(Y,Z)>l)\,dl 
   = \int_{0}^{\Omega_{\max}}\![1-F_Y(l)F_Z(l)]\,dl,
\)
where the limits truncate at $\Omega_{\max}=\max(\Omega_{\ell},\Omega_{k})$.  
Splitting the integration domain into $[0,\Omega_{\min}]$ and $[\Omega_{\min},\Omega_{\max}]$, substituting the Rician CDFs from Lemma~\ref{lem:outage}, and using the Marcum-$Q$ function series expansion in~\cite[Eq.~(4.47)]{alounibook} yield expressions in \eqref{eq_sinr_energy}-\eqref{eq:F_G_defs}.
\end{IEEEproof}
\vspace{-0.2cm}

\begin{figure*}[h]
    \centering
    \subfloat[System outage probability versus threshold~$\zeta$.]{\includegraphics[width=0.33\textwidth]{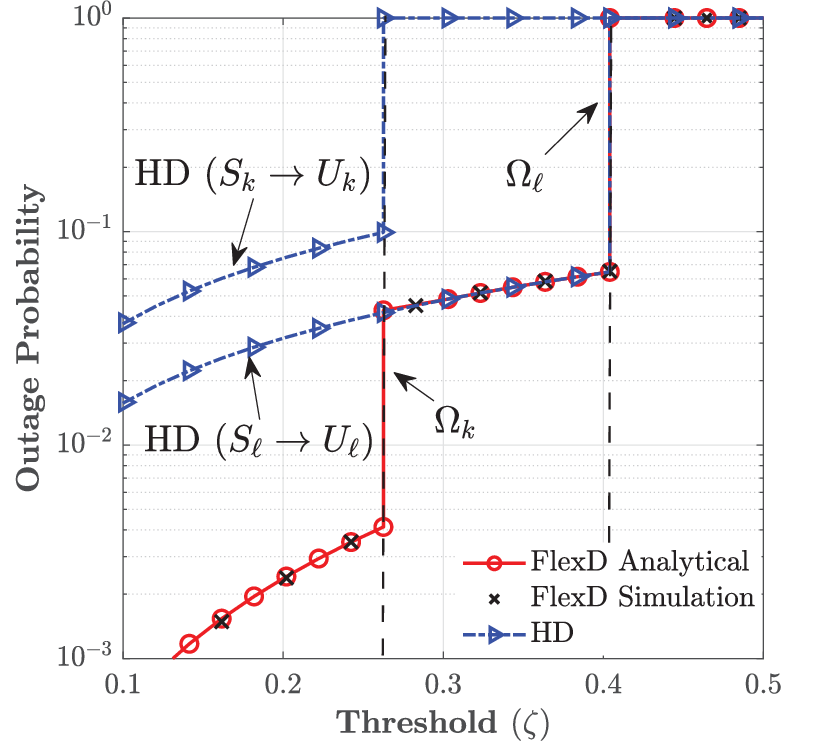}\label{fig:outage_thr}} 
    \hfill
    \subfloat[Energy efficiency versus transmit power~$P$.]{\includegraphics[width=0.33\textwidth]{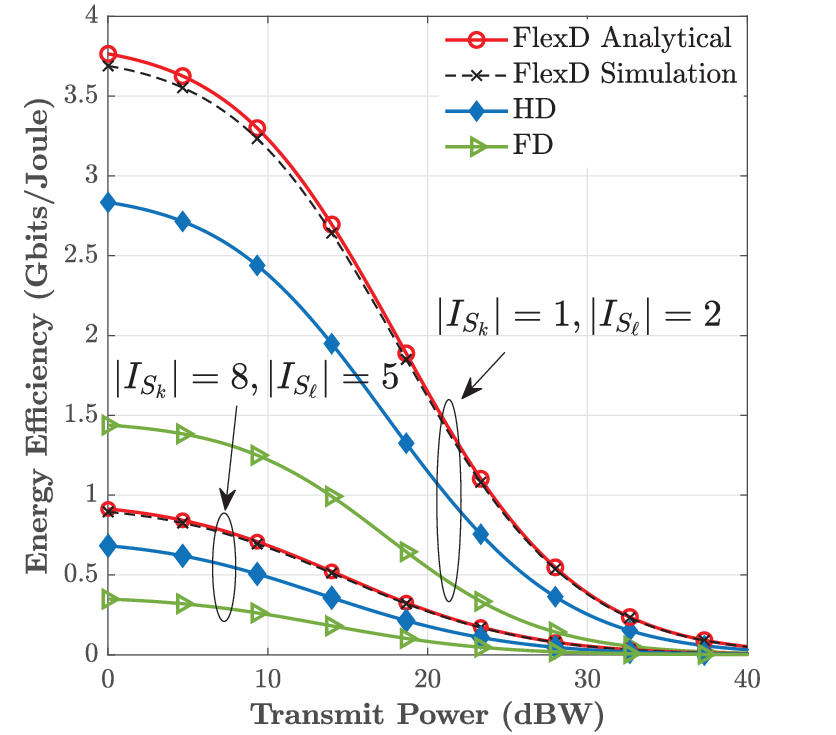}\label{fig_all}} 
    \hfill
    \subfloat[Temporal variation of EE over time slots~$t$.]{\includegraphics[width=0.33\textwidth]{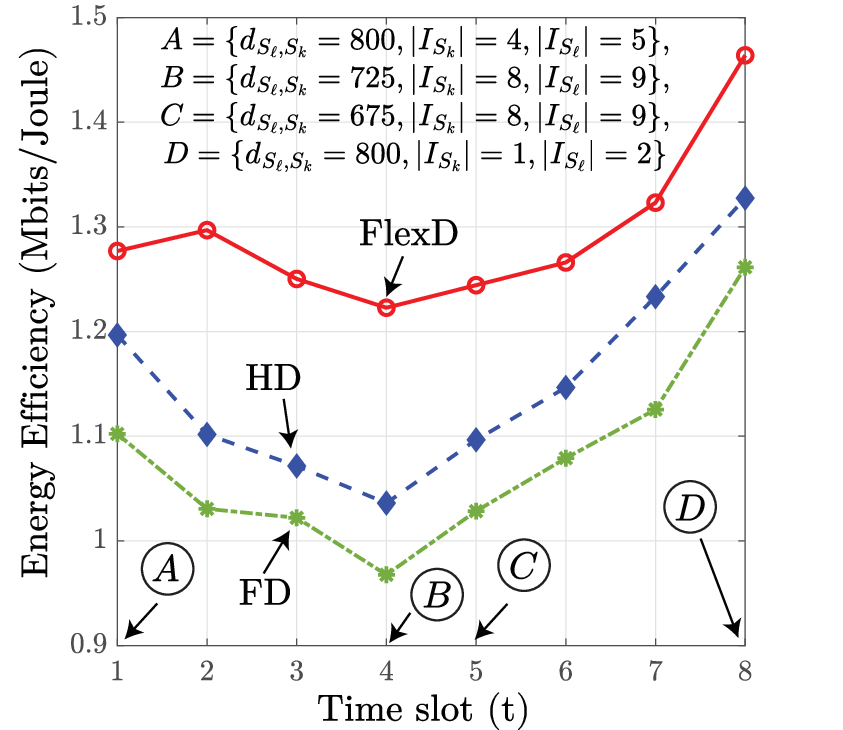}\label{fig:optim}} 
    \caption{Performance comparison of FlexD with HD and FD baselines. 
(a) Lower outage  via adaptive direction selection, 
(b) higher energy efficiency per unit power, and 
(c) robustness to dynamic channel and interference variations across slots.}
    \vspace{-4mm}
\end{figure*}

\subsection{Extension to NLoS, Multiuser, and Multihop Scenarios}

The proposed analysis is naturally extendable to non-line-of-sight (NLoS), multiuser, and multihop configurations.  

In \textit{NLoS conditions}, the SGL may follow a Nakagami-$m$ fading instead of Rician fading.  
Accordingly, the CDF $F_{X}(\cdot,\cdot)$ in Lemma~\ref{lem:outage} and~\eqref{eq:Rerg} is replaced by the Nakagami-$m$ form  
\(
F_X(x)=1-\exp\!\left(-{m x}/{\bar{\gamma}}\right)
\sum_{r=0}^{m-1}{(m x/\bar{\gamma})^{r}}/{r!},
\)
while all subsequent derivations remain structurally identical.

For the \textit{multiuser case}, let $\Gamma_{mu}\!\triangleq\!\max_{i\in\mathcal{U}}\Gamma_{i}$ denote the instantaneous SINR of the best user link among the set of served users $\mathcal{U}$.  
The corresponding outage probability extends from~\eqref{eq:CDF_gamma_star} as  
\(P_{\text{o,mu}}(\zeta)=\Pr(\Gamma_{mu}\!\le\!\zeta)
=\prod_{i\in\mathcal{U}}F_{\Gamma_{u}}(\zeta)\),
assuming independent user links  with CDF for each $F_{\Gamma_{u}}(\zeta)$.

For the \textit{multihop case}, where a packet traverses $H$ serial ISLs before reaching its serving satellite, the end-to-end SINR becomes  
\(\Gamma_{\text{mh}}=\min_{h\in\{1,\dots,H\}}\Gamma_{h}\),
capturing the bottleneck of the weakest hop.  
Its CDF is given by  
\(F_{\Gamma_{\text{mh}}}(\zeta)
=1-\prod_{h=1}^{H}\!\big(1-F_{\Gamma_{h}}(\zeta)\big)\),
from which the corresponding outage and ergodic-rate expressions follow directly.

\section{Numerical Analysis}
Numerical simulations are conducted for LEO satellites operating at altitudes of $500$-$1200$~km~\cite{9632432} with orbital inclinations of $50^{\circ}$-$98^{\circ}$. The inter-satellite distance is $600$-$1200$~km, corresponding to angular separations of $4.5^{\circ}$-$10^{\circ}$ for circular orbits~\cite{radhakrishnan2016survey}. The RF-ISL operates at $f=25$~GHz (Ka-band) with a $500$~MHz bandwidth. The SGL channels $h_{S_k,U_\ell}$ and $h_{S_\ell,U_k}$ follow Rician fading with $|\mu|=1.56$ and $\sigma_g^2=1.3$. Receiver noise power is set to $-115$~dBm, and antenna gains range from $35$ to $40$~dBi.
\subsubsection{System Outage Performance}
Fig.~\ref{fig:outage_thr} shows the system outage probability versus threshold $\zeta$ with $(|\mathcal{I}_{S_k}|, |\mathcal{I}_{S_\ell}|) = (5,8)$, highlighting the reliability gain of the proposed FlexD scheme over conventional HD operation. The analytical results derived for FlexD in~\eqref{eq:CDF_gamma_star} closely match Monte Carlo simulations, validating the analysis. From a communication-theoretic perspective, the FlexD outage behavior arises from the composite bottleneck structure $\gamma_{S_\ell\to U_\ell}$ in~\eqref{eq:rate_to_sinr_equiv1} and $\gamma_{S_k\to U_k}$, where the instantaneous ISL SINRs $\gamma_{S_\ell\to S_k}$ and $\gamma_{S_k\to S_\ell}$ jointly interact with the backlog-limited terms $\gamma_{Q_{S_\ell\to U_\ell}}^{(W)}$ and $\gamma_{Q_{S_k\to U_k}}^{(W)}$. The two transition thresholds $\Omega_{k}$ and $\Omega_{\ell}$ partition the outage curve into three regions corresponding to distinct dominant constraints: (i) ISL-limited, (ii) backlog-limited, and (iii) downlink-limited operation. These regions collectively characterize how link quality and data availability couple to determine the end-to-end reliability.  
As observed, the FlexD curve exhibits continuous transitions across these regions, indicating adaptive rate balancing under moderate channel and buffer variations. In contrast, the HD curves-fixed to either $S_k\!\to\!U_k$ or $S_\ell\!\to\!U_\ell$ per slot-show distinct discontinuities at their respective thresholds, reflecting single-direction bottlenecks. FlexD’s consistent outage reduction highlights its adaptive channel utilization, maximizing mutual information over the two-hop chain and ensuring higher reliability than non-adaptive HD operation.

\subsubsection{Energy Efficiency Performance}
Fig.~\ref{fig_all} depicts the variation of energy efficiency (EE), measured in~Gbits/Joule, versus the total transmit power $P$ (dBW). The analytically derived upper-bounded EE from \eqref{eq:Rbar_approx} with $M\!\ge\!20$ (truncation limit of the infinite series) closely match simulations, and the actual EE curve without Jensen’s approximation is used for FlexD simulations.
For fair comparison, both HD and FlexD modes operate with identical transmit power~$P$ per slot, while  FD employs~$2P$ due to simultaneous bidirectional transmission. The FD mode further experiences residual self-interference (RSI) fixed at~$-120$~dBm-near the thermal noise floor-representing ideal FD conditions. We consider two interference configurations at $S_k$ and $S_\ell$: 
$(|\mathcal{I}_{S_k}|, |\mathcal{I}_{S_\ell}|) = (1,2)$ and $(8,5)$. At slot $t$, the configuration with fewer interferers achieves higher EE due to increased throughput, and vice versa. The results demonstrate that FlexD consistently achieves higher EE than both HD and FD. This gain stems from its direction-adaptive transmission, which maximizes the instantaneous throughput per unit energy. At $P=10$~dBW, with $(|\mathcal{I}_{S_k}|, |\mathcal{I}_{S_\ell}|) = (1,2)$ the EE values for FlexD, HD, and FD are approximately $3.2$, $2.4$, and $1.25$~Gbits/Joule, respectively, highlighting that FlexD can substantially improve power efficiency under energy-limited LEO satellite operations.

\subsubsection{Impact of Satellite Mobility on Energy Efficiency}
Fig.~\ref{fig:optim} shows the evolution of EE across time slots under dynamic conditions induced by satellite motion. Each slot experiences distinct EE values due to time-varying path loss and interference. For example, points~A and~D exhibit identical inter-satellite distances $d_{S_\ell,S_k}$, yet point~D achieves higher EE owing to reduced co-channel interference, demonstrating that interference power dominates when geometric losses are similar. Conversely, point~C yields higher EE than point~B despite equal interference levels, indicating that shorter $d_{S_\ell,S_k}$ (stronger LoS gain) dominates SINR performance.  

\begin{remark}[Energy-Efficiency Adaptability of FlexD]
The FlexD scheme consistently achieves higher EE across all time slots compared to HD and FD. Its robustness stems from slot-wise adaptation to instantaneous channel and interference conditions, effectively maximizing mutual information per joule under varying free-space path loss (FSPL) and inter-satellite interference (ISI).
\end{remark}

\section{Conclusion}
\label{sec:conclusion}

This paper presented a \emph{Flexible Duplex (FlexD)} transmission framework for low Earth orbit (LEO) satellite networks equipped with inter-satellite links (ISLs). The proposed scheme addresses the directional asymmetry inherent in half-duplex ISLs and the backlog-induced two-hop relaying problem arising from satellite mobility and fading. By formulating the ISL-downlink bottleneck in a unified SINR domain, closed-form expressions for throughput outage and energy efficiency were derived under deterministic ISLs and Rician satellite-ground channels. Analytical and Monte Carlo results showed close agreement, revealing distinct operating regimes driven by ISL quality and queue backlogs. Compared with conventional half- and full-duplex modes, FlexD achieves significantly lower outage probability and up to 30\% higher energy efficiency, attributed to its slot-wise direction adaptation 
under time-varying interference and path loss conditions. Future extensions include expanding the framework to multi-hop constellations with dynamic topology and utilizing learning-based methods to enable scalable duplex management in next-generation space networks.


\end{document}